\documentclass[11pt]{article}
\pdfoutput=1

\newif\ifcount

\usepackage[utf8]{inputenc}
\usepackage[T1]{fontenc}
\usepackage[english]{babel}

\usepackage{mathtools}
\usepackage{amssymb}
\usepackage{bm}
\usepackage{bbm}
\usepackage{mleftright}
\usepackage{multirow}
\usepackage{booktabs}

\usepackage{tikz}
\usetikzlibrary{backgrounds,fit,decorations.pathreplacing,calc}
\usepackage{quantikz}

\usepackage{amsthm}
\theoremstyle{plain}
\usepackage{thmtools}
\usepackage{thm-restate}

\usepackage{caption}
\usepackage{subcaption}

\usepackage{verbatim}
\usepackage{enumitem}
\usepackage{algorithm}
\usepackage{algcompatible}
\usepackage{float}
\usepackage{fullpage}
\usepackage{titling}
\usepackage{csquotes}

\usepackage[final=true,colorlinks=true,allcolors=blue]{hyperref}

\makeatletter
\def\namedlabel#1#2{\begingroup
  #2%
  \def\@currentlabel{#2}%
  \phantomsection\label{#1}\endgroup
}
\makeatother
\usepackage[backend=biber,style=numeric-comp,sorting=none,backref=true,maxbibnames=99]{biblatex}
\newcommand{\Imid}{\mathcal{I}_{\mathrm{mid}}}
\newcommand{\uu}{\mathbf{u}}
\newcommand{\vv}{\mathbf{v}}

\newcommand{\diag}[1]{\mathrm{diag}\left( #1 \right)}

\newcommand{\C}{\mathbb{C}}

\newcommand{\R}{\mathbb{R}}

\newcommand{\pvp}{\vec{p}{\kern 0.45mm}'}
\let\oldnabla\nabla
\renewcommand{\nabla}{\oldnabla\!}

\newcommand{\underflow}[2]{\underset{\kern-60mm \overbrace{#1} \kern-60mm}{#2}}

\long\def\ignore#1{}

\newtheorem{theorem}{Theorem}
\newtheorem{corollary}[theorem]{Corollary}
\newtheorem{lemma}[theorem]{Lemma}

\newtheorem*{claim*}{Claim}
\newtheorem{remark}[theorem]{Remark}

\usetikzlibrary{backgrounds,fit,decorations.pathreplacing,calc}

\title{	
	Explicit Quantum Circuits for Simulating Linear Differential Equations via Dilation
}
\author{
	Seonggeun Park\thanks{Department of Electrical Engineering, Korea University, South Korea. \texttt{ssgg0926@korea.ac.kr} }
}
\date{October 3, 2025\vspace{-5mm}}

\begin{document}
	
\maketitle
 
\begin{abstract}
Quantum simulation has primarily focused on unitary dynamics, while many physical and engineering systems can be modeled by linear ordinary differential equations whose generators include non-Hermitian terms. Recent studies have shown that such equations, which give rise to nonunitary dynamics, can be embedded into a larger unitary framework via dilation techniques. However, their concrete realization on quantum circuits remains underexplored.

In this paper we present a concrete pipeline that connects the dilation formalism with explicit quantum circuit constructions. On the analytical side, building on the recent dilation framework, we introduce a discretization of the continuous dilation operator that is tailored for quantum implementation. This construction ensures an exactly skew-Hermitian ancillary generator, which allows the moment conditions to be satisfied without imposing artificial constraints. We prove that the resulting scheme achieves a global error bound of order $O(M^{-3/2})$, up to exponentially small boundary effects. This error can be suppressed by refining the discretization, where $M$ denotes the discretization parameter.

On the algorithmic side, we demonstrate that the dilation triple $(F_h, |r_h\rangle, \langle l_h|)$ can be efficiently implemented on quantum circuits. Using linear combinations of unitaries, QFT-adder operators, and quantum singular value transformation, the framework requires resources ranging from $O(\log M)$ to $O((\log M)^2)$, depending on the stage of the pipeline.
\end{abstract}

\section{Introduction}

Many physical and engineering systems can be modeled by linear differential equations whose generators are not necessarily Hermitian, for example due to dissipation, relaxation, or coupling to an external environment. Simulating such equations is relevant for understanding phenomena such as open quantum systems, transport processes, and dissipative dynamics. We focus on simulating homogeneous linear differential equations of the form
\begin{equation}
    \dot{\mathbf{x}}(t) = A(t)\mathbf{x}(t),
    \label{eq:ODE}
\end{equation}
where $A(t)\in\C^{N\times N}$ and $\mathbf{x}(0)=\mathbf{x}_0\in\C^N$. The generator can always be decomposed as
\begin{equation}
    A(t) = -iH(t)+K(t),
    \label{eq:A-decomposition}
\end{equation}
with Hermitian operators $H$ and $K$. When $K=0$, the setting reduces to the Hamiltonian simulation problem. The more general case $K\neq 0$ corresponds to non-Hermitian generators. Equation~\eqref{eq:ODE} already covers many important models; in particular, after spatial discretization\footnote{Spatial discretization refers to approximating spatial derivatives by finite differences or related numerical schemes, resulting in a matrix $A(t)$ that encodes the corresponding spatial operators.
}, a broad range of PDEs can be reformulated as large-scale ODE systems of this form.

In practical applications, the dimension $N$ can be extremely large, especially when $A(t)$ arises from the discretization of partial differential equations in multiple spatial dimensions. This leads to substantial computational challenges for classical algorithms in simulating such large differential equation systems.

Quantum computers offer a potential route to overcome these limitations. Over the past decade, there has been remarkable progress in Hamiltonian simulation algorithms~\cite{berry2007efficient,10.5555/2231036.2231040,10.5555/2481569.2481570,10.5555/2600498.2600499,berry2014exponential,berry2015simulating,berry2015hamiltonian,low2017optimal,childs2018toward,childs2019faster,low2019hamiltonian}, achieving asymptotically optimal complexity in simulation time and precision. In contrast, simulating linear dynamics with non-Hermitian generators, which correspond to nonunitary evolution, remains more challenging. A quantum device can, in principle, address this issue by embedding the evolution into a larger Hilbert space where the dynamics become unitary, a technique referred to as \emph{dilation}.

Two major dilation approaches have emerged, Schrödingerization~\cite{jin2023quantum,jin2024quantum} and the Linear Combination of Hamiltonian Simulation (LCHS)~\cite{an2023linear,an2023quantum}. Both provide exact embeddings of nonunitary dynamics into unitary dynamics. More recently, a general dilation framework was proposed~\cite{li2025linear}. This framework unifies Schrödingerization and LCHS, and further introduces new families of dilations through integral kernels, difference operators, and a pseudodifferential generator. Despite these advances, prior works have focused primarily on the mathematical structure of dilation and the analysis of discretization errors. In contrast, the concrete realization on quantum circuits, including state preparation, block-encoding of dilated Hamiltonians, and evaluation strategies, remains comparatively underexplored. Consequently, the connection between mathematical dilation methods and practical quantum algorithms has not yet been fully established.

In this work, we aim to bridge this gap. Building on the dilation framework of~\cite{li2025linear}, we adapt the Summation-by-Parts (SBP) discretization to define a skew-Hermitian operator $F_h$ that preserves the correct algebraic structure while being well-suited for block encoding on quantum hardware. Based on this operator, we provide an error analysis and establish a global error bound. We further present explicit circuit constructions, including preparation of the ancillary state $\ket{r_h}$ via Quantum Singular Value Transformation (QSVT)~\cite{gilyen2019quantum,martyn2021grand}, block encoding of the dilated Hamiltonian, and an evaluation strategy.

In summary, the main contribution of this paper is to integrate these components into a unified framework for simulating linear ordinary differential equations with non-Hermitian generators. Section~\ref{sec:math-framework} introduces the mathematical formulation of the dilation technique and the SBP-based operator $F_h$. Section~\ref{sec:error-analysis} develops the error analysis and establishes the global bound. Section~\ref{sec:circuit} presents the quantum circuit implementations. This pipeline connects abstract dilation methods to practical quantum algorithms, with potential applications to physically relevant PDE systems such as viscoelastic wave and heat equations.

\section{Mathematical Framework}
\label{sec:math-framework}

\subsection*{Notation and conventions}

For a function $u=u(t,p)$ we use the shorthand
\[
u_t := \partial_t u, \qquad u_p := \partial_p u,
\]
so that subscripts on functions denote partial derivatives. 

Boldface letters such as $\uu=(u_0,\dots,u_M)^\top \in \C^{M+1}$ denote vectors, with subscripts indicating components. The individual components $u_i$ are written in normal (non-bold) font. When a quantity depends on time, we use a dot to denote its time derivative, e.g.\ $\dot{\uu}(t):=\partial_t \uu(t)$ for a vector and $\dot{u}_i(t):=\partial_t u_i(t)$ for its components. We denote by $\mathbf{e}_M\in\C^{M+1}$ the $M$-th standard basis vector, whose last component is $1$ and all others are $0$.

An $n$-bit Toffoli gate refers to the multi-controlled NOT gate with $n$ control qubits and one target qubit, hence acting on a total of $n{+}1$ qubits. When we say ``Toffoli gate'' without qualification, we mean the standard 2-bit Toffoli gate with two controls and one target.

In this paper, we only consider the case where $K(t)$ is negative semidefinite.

\subsection{Dilation framework and moment conditions}\label{subsec:dilation}

In~\cite{li2025linear}, a general framework for embedding linear dynamics with non-Hermitian generators into unitary evolution is formalized by the following theorem.

\begin{theorem}[Moment conditions for exact dilation, {\cite[Theorem~1]{li2025linear}}]
Let $F$ be a linear operator acting on the ancillary Hilbert space $\mathcal{H}_A$, and let $|r) \in \mathbb{X}$ together with a linear functional $(l|$ on $\mathbb{X}$, where $\mathcal{H}_A \subset \mathbb{X}$. Let $H(s)$ and $K(s)$ be Hermitian operators on the system Hilbert space $\mathcal{H}_S$, defined for $s\in[0,T]$, where $T>0$ denotes the final time. If the following moment conditions are satisfied,
\begin{equation}
    (l|F^k|r) = 1, \qquad \forall k \ge 0,
    \label{eq:moment-condition}
\end{equation}
then the dilated evolution reproduces the exact solution of the target dynamics:
\begin{equation}
    \big((l|\otimes I\big)\,\mathcal{T}\exp\!\left(-i \int_0^t \big(I_A \otimes H(s) + i F \otimes K(s)\big)\, ds\right) \big(|r)\otimes I\big) = \mathcal{T}\exp\!\left(\int_0^t A(s)\, ds\right),
    \label{eq:dilation-method}
\end{equation}
where $A(s) = -iH(s) + K(s)$.
\end{theorem}

This result states that if a triple $(F,|r),(l|)$ satisfies the moment condition~\eqref{eq:moment-condition}, then the dilated unitary evolution exactly recovers the physical solution of the linear ODE~\eqref{eq:ODE}. Different choices of the triple $(F,|r),(l|)$ correspond to different dilation methods; in particular, Schrödingerization and LCHS can be viewed as special cases within this general framework.

\begin{remark}[On the notation $(l|,|r)$]
The symbols $(l|$ and $|r)$ are written in round brackets to emphasize that they do not belong to the Hilbert space $\mathcal{H}_A$. Here $|r)$ belongs to a larger space $\mathbb{X}\supset\mathcal{H}_A$, while $(l|$ is a linear functional on $\mathbb{X}$.
\end{remark}

The need for such a dilation arises from the impossibility of performing direct Hamiltonian simulation within the system space $\mathcal{H}_S$, since $H(s) + iK(s)$ is not Hermitian. By introducing an ancillary space and a skew-Hermitian operator $F$, one can lift the dynamics to the enlarged space $\mathcal{H}_A \otimes \mathcal{H}_S$ with Hamiltonian
\[
I_A \otimes H(s) + i F \otimes K(s),
\]
which is Hermitian.

It is crucial that $|r)\in\mathbb{X}$ and $(l|$ acts as a functional on $\mathbb{X}$, not merely on $\mathcal{H}_A$. Otherwise, if both $|r)$ and $(l|$ were restricted to $\mathcal{H}_A$, then the left-hand side of the equation~\eqref{eq:dilation-method} would represent a bounded unitary evolution, while the right-hand side may not be bounded, leading to a contradiction.

In~\cite{li2025linear}, the ancillary Hilbert space was chosen as
\begin{equation}
\mathcal{H}_A := H^1_0(0,1)
= \big\{ f \in L^2(0,1) \mid f' \in L^2(0,1),\, f(1)=0 \big\}, 
\quad
\langle f,g \rangle = \int_0^1 f(p)^{\ast} g(p)\, dp.
\label{eq:HA}
\end{equation}
That is, $\mathcal{H}_A$ consists of functions on $[0,1]$ that are square-integrable together with their derivatives, and that satisfy the boundary condition $f(1)=0$. Within this space, the ancillary operator was defined as
\begin{equation}
F_\theta = \theta F, 
\qquad F := p \partial_p + \tfrac{1}{2}.
\label{eq:Ftheta}
\end{equation}
Finally, the ancillary initial state $|r)$ and evaluation functional $(l|$ were specified by
\begin{equation}
|r) = p^{\beta}, 
\quad \beta = \tfrac{1}{\theta}-\tfrac{1}{2}, 
\qquad
(l|f = 2^{\beta} f\!\left(\tfrac{1}{2}\right).
\label{eq:rstate}
\end{equation}

\begin{lemma}[{\cite[Lemma~1]{li2025linear}}]
For the above choice of triple $(F_\theta,|r),(l|)$, the operator $F$ is skew-Hermitian on $\mathcal{H}_A$, and the moment condition \eqref{eq:moment-condition} is fulfilled.
\end{lemma}

\begin{proof}[Sketch of proof]
For $f,g \in \mathcal{H}_A$,
\[
\langle f, F g \rangle
= \int_0^1 f(p)^{\ast} \Big(p \tfrac{d}{dp} g(p) + \tfrac{1}{2} g(p)\Big)\, dp.
\]
Integration by parts, together with the boundary condition $f(1)=g(1)=0$, yields $\langle f, F g \rangle = - \langle F f, g \rangle$, establishing skew-Hermiticity. Moreover, $F_\theta |r) = |r)$ and $(l|r)=1$, which confirms the moment condition.
\end{proof}

A skew-Hermitian operator has purely imaginary eigenvalues when restricted to the Hilbert space $\mathcal{H}_A$. Thus, notice that the identity $F_\theta|r)=|r)$ holds because $|r)=p^{\beta}$ does not belong to $\mathcal{H}_A$. This interplay between $\mathcal{H}_A$ and its embedding space $\mathbb{X}$ is crucial for the moment conditions to hold. This establishes the continuous dilation framework and clarifies the role of the moment conditions.

\subsection{Discretization of the triple \texorpdfstring{$(F_\theta,|r),(l|)$}{(F,|r),(l|}}\label{subsec:discretization}

While the triple $(F_\theta,|r),(l|)$ in Sec.~\ref{subsec:dilation} is mathematically elegant, it relies on the ancillary Hilbert space $\mathcal{H}_A=H^1_0(0,1)$ and a carefully chosen triple $(F_\theta,|r),(l|)$ such that the moment conditions are exactly satisfied. In this setting, $F$ is skew-Hermitian only when restricted to $\mathcal{H}_A$, and the ancillary initial state $|r)$ lies outside $\mathcal{H}_A$ but within the larger embedding space $\mathbb{X}$. This distinction is essential for the validity of the theorem, but it creates a mismatch when considering implementation on a quantum computer.

Indeed, a gate-based quantum device naturally operates on a finite-dimensional Hilbert space of qubit registers with the standard $\ell^2$ inner product and cannot enforce a restriction such as $f(1)=0$ at the boundary. In~\cite{li2025linear}, discrete approximations compatible with the standard $\ell^2$ inner product were proposed. It was shown that, when restricted to the subspace $\{\uu\in\C^{M+1}\mid u_M=0\}$, a suitably rescaled discrete operator becomes strictly skew-Hermitian. However, this approach requires an explicit boundary constraint as well as a consistent rescaling of the initial state $|r)$ and the evaluation functional $(l|$. This complicates state preparation and undermines the simplicity of the dilation when realized on quantum hardware.

To overcome these limitations, we design a new discretization $F_h$ of the operator $F=p\partial_p+\tfrac{1}{2}$ that is skew-Hermitian on the finite-dimensional $\ell^2$ space without boundary restriction. This eliminates the need for boundary conditions or state rescaling and makes the formulation directly compatible with quantum hardware.

We now describe the construction in detail. We discretize the ancillary coordinate $p\in[0,1]$ into $M$ subintervals with mesh size $h=1/M$ and grid nodes $p_j=jh$ for $j=0,1,\dots,M$. We present the discrete triple
\[
\ket{r_h}\in\C^{M+1},\qquad \bra{l_h}:\C^{M+1}\to\C,\qquad \theta F_h\in\C^{(M+1)\times (M+1)},
\]
which plays the role of the continuous triple $(F_\theta,|r),(l|)$ in \eqref{eq:Ftheta}--\eqref{eq:rstate}.

We begin with $\ket{r_h}$. Discretizing $|r)=p^\beta$ with $\beta=\tfrac{1}{\theta}-\tfrac{1}{2}$ yields
\begin{equation}
  \ket{r_h}
  := C_{M,\theta}^{-1}\sum_{j=0}^{M} p_j^{\beta}\ket{j}, 
  \qquad
  C_{M,\theta}=\Big(\sum_{j=0}^{M} p_j^{2\beta}\Big)^{1/2}.
  \label{eq:rh}
\end{equation}
A bound on the normalization constant is given below.

\begin{lemma}[Bound for $C_{M,\theta}$]\label{lem:CM-bounds}
For $\beta \ge 0$,
\begin{equation}
  \frac{M}{2\beta+1}\;\le\; C_{M,\theta}^2 \;\le\; \frac{M}{2\beta+1} \;+\; 1.
  \label{eq:CM-bound}
\end{equation}
\end{lemma}
\begin{proof}[Proof sketch]
Let $\alpha:=2\beta \ge 0$. Since $x^\alpha$ is increasing for $\alpha\ge0$,
\[
\sum_{j=1}^M j^\alpha \le \int_0^M x^\alpha dx + M^\alpha
= \frac{M^{\alpha+1}}{\alpha+1}+M^\alpha,
\]
and division by $M^\alpha$ gives the upper bound.
The lower bound follows from $\sum_{j=1}^M j^\alpha \ge \int_0^M x^\alpha dx$.
\end{proof}

Next, we turn to the evaluation functional $\bra{l_h}$. Let
\[
\Imid := \{x\in\{0,1,\dots,M\}\mid M/4 \le x \le 3M/4\}.
\]
We define $\bra{l_h}$ so that it corresponds to postselection on ancilla outcomes $x\in\Imid$:
\begin{equation}
  \bra{l_h} := C_{M,\theta}\Big(\frac{M}{x}\Big)^{\beta}\bra{x},
  \qquad x\in \Imid.
  \label{eq:lh}
\end{equation}
It immediately follows that $\langle l_h\vert r_h\rangle = 1$ for all $x\in\Imid$.

The operator $F$ is discretized based on the SBP framework. Note that $F=\tfrac12\{\partial_p,p\}$. Let $P=\diag{0,h,2h,\dots,1}$ and adopt the second-order SBP difference operator $D\in\C^{(M+1)\times(M+1)}$. For vectors $\uu,\vv\in\C^{M+1}$ with
$\uu=(u_0,u_1,\dots,u_M)^\top$, $\vv=(v_0,v_1,\dots,v_M)^\top$, set
\begin{equation}
D=\frac{1}{h}\begin{bmatrix}
-1 & 1 & & & \\
-\tfrac12 & 0 & \tfrac12 & & \\
& \ddots & \ddots & \ddots & \\
& & -\tfrac12 & 0 & \tfrac12 \\
& & & -1 & 1
\end{bmatrix},\qquad
H=h\cdot\diag{\tfrac12,1,\dots,1,\tfrac12},
\label{eq:SBP-DH}
\end{equation}
which satisfies the SBP identity:
\begin{equation}
  \langle \uu,D\vv\rangle_H+\langle D\uu,\vv\rangle_H
  = \uu^\dagger B \vv,\quad
  \langle \uu,\vv\rangle_H:=\uu^\dagger H \vv,\quad
  B:=\diag{-1,0,\dots,0,1}.
  \label{eq:SBP-identity}
\end{equation}
We then form the split operator
\begin{equation}
  G_h:=\tfrac12(PD+DP),
  \label{eq:Gh-def}
\end{equation}
which is skew-Hermitian in the $H$-inner product only on the subspace $\{\uu\in\C^{M+1}:u_M=0\}$.

The next step is to construct an operator that is skew-Hermitian in the $H$-inner product, but without imposing the constraint $u_M=0$. Let $-\kappa(t)$ be regarded as the spectrum of the negative semidefinite $K(t)$. Then, the ancilla dynamics can be represented as
\begin{equation}
  \dot\uu(t) = -\theta\,\kappa(t)\, G_h\,\uu(t),\qquad \kappa(t)>0.
  \label{eq:ancilla-ode-nosat}
\end{equation}
To remove the boundary constraint while preserving skewness, we introduce a Simultaneous Approximation Term (SAT) at the right boundary~\cite{carpenter1994time,fernandez2014review}:
\begin{equation}
  \dot \uu(t)= -\theta\,\kappa(t)\, G_h \uu(t) + \tau_M\, H^{-1}\big(u_M(t)-g_R(t)\big)\,\mathbf{e}_M,
  \label{eq:semi-discrete-SAT}
\end{equation}
where $\mathbf{e}_M=(0,0,\dots,0,1)^\top\in\mathbb{C}^{M+1}$ denotes the $M$-th standard basis vector. Define the energy $E(t):=\uu(t)^\dagger H\,\uu(t)$. Differentiating and applying the SBP identity yields
\begin{equation}
  \frac{d}{dt}\big(\uu^\dagger H\,\uu\big)
  = \theta\,\kappa(t)\,|u_M(t)|^2 + 2\tau_M\,\Re\big(u_M^\ast(t)\,(u_M(t)-g_R(t))\big).
  \label{eq:energy-balance}
\end{equation}
Choosing
\begin{equation}
  \tau_M=\frac{\theta\,\kappa(t)}{2}
  \label{eq:tau-choice}
\end{equation}
cancels the $|u_M|^2$ term and gives $\tfrac{d}{dt}(\uu^\dagger H\uu)=-\theta\,\kappa(t)\,\Re(u_M^\ast g_R)$, which shows conservation in the $H$-inner product when $g_R= 0$. Thus, \eqref{eq:semi-discrete-SAT} can be equivalently written as
\begin{equation}
  \dot \uu(t) = -\theta\,\kappa(t)\,\widetilde G_h\,\uu(t),\qquad
  \widetilde G_h:= G_h - \tfrac12 H^{-1}\mathbf{e}_M \mathbf{e}_M^\top,
  \label{eq:Gh-tilde}
\end{equation}
which is $H$-skew on all of $\C^{M+1}$.

Finally, since quantum hardware natively operates with the standard $\ell^2$ inner product, it is necessary to transform an operator that is skew-Hermitian with respect to the $H$-inner product into one that is skew-Hermitian in the $\ell^2$ inner product. This is achieved by applying a similarity transformation with $H^{1/2}$:
\begin{equation}
  \widetilde F_h := H^{1/2}\,\widetilde G_h\,H^{-1/2},
  \label{eq:Fh-tilde}
\end{equation}
which is exactly skew-Hermitian in $\ell^2$ and remains tridiagonal. For implementation, we adopt the simplified stencil
\begin{equation}
F_h=\frac{1}{4}\begin{bmatrix}
0 & 1 & 0 & 0 &  & 0 & 0 & 0 \\
-1 & 0 & 3 & 0 & \cdots & 0 & 0 & 0 \\
0 & -3 & 0 & 5 &  & 0 & 0 & 0 \\
  & \vdots & \vdots &  & \ddots &  & \vdots & \\
0 & 0 & 0 & 0 &\cdots & -(2M-3) & 0 & (2M-1) \\
0 & 0 & 0 & 0 & \cdots & 0 & -(2M-1) & 0
\end{bmatrix}.
\label{eq:Fh}
\end{equation}
Here $\Delta_h := F_h - \widetilde F_h$ is skew-Hermitian and supported only on the two $2\times2$ corners. Accordingly, $\widetilde F_h$ is replaced by $F_h$, which incurs only a negligible error. This discrepancy is part of the boundary mismatch error, whose impact will be analyzed in Section~\ref{sec:error-analysis}. Moreover, $F_h$ is particularly well suited for block-encoding, as will be demonstrated in Section~\ref{sec:block-encode-Fh-and-assemble}.


\section{Error Analysis and Global Bound}
\label{sec:error-analysis}

In this section, we quantify the error incurred by replacing the continuous dilation generator 
$F=\tfrac12\{\partial_p,p\}=p\partial_p+\tfrac12$ with $F_h$, together with the boundary modifications used for block-encoding. 
Throughout, let $g(p)=p^\beta$ with $\beta=\tfrac{1}{\theta}-\tfrac{1}{2}$ and
\[
\bm{g}=(g(p_0),\dots,g(p_M))^\top\in\C^{M+1},\qquad
p_j=jh,\quad h=\tfrac{1}{M}.
\]
For the discrete evaluation we recall
\[
\ket{r_h}=C_{M,\theta}^{-1}\sum_{j=0}^{M}p_j^\beta\ket{j},
\qquad 
\bra{l_h}=C_{M,\theta}\Big(\tfrac{M}{x}\Big)^\beta\bra{x},
\qquad x\in\Imid,
\]
so that $\langle l_h\vert r_h\rangle=1$. We assume $K(t)\preceq 0$ and $\|K(t)\|\le K_{\max}$ on $t\in[0,T]$.

\subsection{Sources of error and the main theorem}

The error relative to the continuous, exact dilation arises from two main sources. First, the SBP-based interior discretization of the dilation generator $F$ incurs a consistency error of order $O(h^2)$ in the interior. The split form $G_h=\tfrac12(PD+DP)$ provides a second-order accurate approximation in the interior, but it cannot perfectly reproduce the action of the continuous operator. Second, the operator $G_h$ is skew-Hermitian only on the restricted subspace $\{u \in \mathbb{C}^{M+1} \mid u_M=0\}$. To extend skew-Hermiticity to the full space $\mathbb{C}^{M+1}$, we add a Simultaneous Approximation Term (SAT) at the right boundary, which cancels the boundary leakage and yields the modified operator $\widetilde G_h$. This boundary modification ensures skew-Hermiticity on the entire space, while introducing an additional boundary error. In the following, we establish a global bound by combining these two contributions into a unified error analysis.

\begin{theorem}[Global error at mid-indices]
\label{thm:main}
Assume $0<\theta\le 2/7$, $\beta=\tfrac{1}{\theta}-\tfrac{1}{2}\ge 3$, and $\theta K_{\max}T\le\tfrac{1}{8e}$. 
Fix $x\in\Imid$. Then, for all unit $\ket{x_0}\in\C^N$,
\begin{align}
\Big\|
\mathcal{T}e^{\int_0^T A(s)\,ds}\ket{x_0}
-
\big(\bra{l_h}\otimes I\big)U_E(T,0)\big(\ket{r_h}\otimes\ket{x_0}\big)
\Big\|
&\le 
\Big(\tfrac{M}{x}\Big)^{\beta}\,
\big|u(T,p_x)-u_x^d(T)\big|,
\label{eq:final-topline}
\\[-1pt]
&\le 4^{\beta}\left(
C(\theta)\,h^{3/2}
+\Big(2+\tfrac{M(1+C(\theta)h^{3/2})}{8e}\Big)2^{-M/4}
\right),
\label{eq:main-bound}
\end{align}
where $U_E(T,0)$ denotes the unitary operator of the dilated evolution
\[
U_E(T,0) = \mathcal{T}\exp\!\left(-i \int_0^T \big(I\otimes H(s) + i \theta F_h\otimes K(s)\big)\, ds\right),
\]
the \emph{ideal} ancilla profile $u$ solves
\begin{align}
u_t(t,p)=-\theta\kappa(t)\,F u(t,p),\qquad u(0,p)=g(p),
\label{eq:ideal-pde}
\end{align}
and the implemented discrete profile solves
\begin{align}
\dot{\bm{u}}^d(t)=-\theta\kappa(t)\,F_h\,\bm{u}^d(t),\qquad \bm{u}^d(0)=\bm{g}.
\label{eq:ud-ode}
\end{align}
Here 
\(
C(\theta):=\frac{\theta}{12}\,\beta(\beta-1)(2\beta-1).
\)
Consequently,
\[
\Big\|
\mathcal{T}e^{\int_0^T A(s)\,ds}\ket{x_0}
-
\big(\bra{l_h}\otimes I\big)U_E(T,0)\big(\ket{r_h}\otimes\ket{x_0}\big)
\Big\|
=O\!\Big(M^{-3/2}+M\,2^{-M/4}\Big).
\]
\end{theorem}

\subsection{Proof of Theorem~\ref{thm:main}}

\begin{lemma}[Solution of the exact dilation]\label{lem:kappa}
The solution of \eqref{eq:ideal-pde} admits the separated form
\[
u(t,p)=y(t)\,g(p),
\]
where $y(t)$ satisfies
\[
\dot{y}(t)=-\kappa(t)\,y(t),\qquad y(0)=1.
\]
If $\kappa(t)\ge0$ for $t\in[0,T]$, then
\[
0< y(t)=\exp\!\Big(-\int_0^t \kappa(s)\,ds\Big)\le 1.
\]
\end{lemma}

\begin{proof}
This follows by direct verification. Substituting the ansatz $u(t,p)=y(t)g(p)$ into
\eqref{eq:ideal-pde} and using $\theta Fg=g$, one obtains the reduced ODE
$\dot{y}(t)=-\kappa(t)y(t)$ with initial condition $y(0)=1$. The stated exponential
representation is immediate.
\end{proof}

To prove Theorem~\ref{thm:main}, we compare the discrete solution $\bm{u}^d(t)$ with a 
\emph{reference solution with enforced boundary value} $\bm{u}^{\mathrm{ex}}(t)\in\C^{M+1}$, which evolves under the same discrete operator $F_h$ but whose right boundary is fixed to match the exact continuous dilation:
\begin{align}
\dot{\bm{u}}^{\mathrm{ex}}(t)=-\theta\kappa(t)F_h\bm{u}^{\mathrm{ex}}(t),\qquad 
\bm{u}^{\mathrm{ex}}(0)=\bm{g},\qquad 
u^{\mathrm{ex}}_M(t)=u(t,1)=y(t),
\label{eq:uex-ode}
\end{align}
where $y(t)=\exp\!\Big(-\int_0^t \kappa(s)\,ds\Big)$. The total error at any node $p_i$ can be decomposed into two contributions:
\begin{align}
\big|u(T,p_i)-u_i^d(T)\big|
\ \le\
\underbrace{\big|u(T,p_i)-u_i^{\mathrm{ex}}(T)\big|}_{\text{discretization error}}
\ +\
\underbrace{\big|u_i^{\mathrm{ex}}(T)-u_i^d(T)\big|}_{\text{boundary-mismatch error}}.
\label{eq:triangle}
\end{align}
In the following, we analyze these two errors separately and then combine them to complete the proof of Theorem~\ref{thm:main}.

\begin{lemma}[Second-order interior error {\cite[Lemma~4]{li2025linear}}]\label{lem:interior}
Let $v\in C^3[0,1]$, and let $\bm{v}\in \mathbb{R}^{M+1}$ with entries $v_i=v(p_i)$. Define
\[
M_3:=\max_{0\le p \le 1}\Big|\,v''(p)+\tfrac{2p}{3}v'''(p)\,\Big|.
\]
Then, for $1\le i\le M-1$,
\[
\big|(F_h \bm{v})_i-(Fv)(p_i)\big|\ \le\ \tfrac{1}{4}\,h^2 M_3,
\]
and if $v(0)=v'(0)=0$ then $(F_h \bm{v})_0 - (Fv)(p_0)=O(h^2)$.
\end{lemma}

\begin{proof}
For $1\le i \le M-1$, the interior stencil is
\[
(F_h \bm{v})_i=\frac{p_{i+1}+p_i}{4h}\,v_{i+1}-\frac{p_{i}+p_{i-1}}{4h}\,v_{i-1}.
\]
Using $p_{i\pm 1}=p_i \pm h$, we obtain the Taylor expansions
\[
v_{i\pm 1}=v(p_i)\pm h v'(p_i)+\tfrac{h^2}{2}v''(p_i)\pm\tfrac{h^3}{6}v'''(\xi_{i\pm}),
\qquad \xi_{i\pm}\in(p_{i-1},p_{i+1}).
\]
Substituting these into the stencil and applying the intermediate value theorem, we find
\[
(F_h\bm{v})_i=\tfrac12 v(p_i)+p_i v'(p_i)+\tfrac{h^2}{4}v''(p_i)+\tfrac{p_i h^2}{6}v'''(\xi_i),
\]
for some $\xi_i\in(p_{i-1}, p_{i+1})$. Thus, compared to $(Fv)(p_i)=\tfrac{1}{2}v(p_i)+p_iv'(p_i)$,
the defect is $\tfrac{h^2}{4}v''(p_i)+\tfrac{p_i h^2}{6}v'''(\xi_i)$.

For the boundary node $i=0$, if $v(0)=v'(0)=0$, then
\[
(F_h \bm{v})_0=\tfrac{1}{2}v_1
=\tfrac{1}{2}\Big(\tfrac{h^2}{2}v''(p_0)+\tfrac{h^3}{6}v'''(\xi_0)\Big)=O(h^2),
\]
for some $\xi_0\in(0,p_{1})$.
\end{proof}

\begin{corollary}\label{cor:g}
If $0< \theta\le \tfrac{2}{7}$ so that $\beta=\tfrac{1}{\theta}-\tfrac{1}{2}\geq 3$, then for $1\le i \le M-1$,
\[
\theta \,\big|(F_h \bm{g})_i - (Fg)(p_i)\big|
\ \le\ C(\theta)\,h^2, 
\qquad 
C(\theta)=\tfrac{\theta}{12}\,\beta(\beta-1)(2\beta-1).
\]
\end{corollary}

\begin{proof}
Apply Lemma~\ref{lem:interior} with $v=g$, where $g(p)=p^\beta$.  
For $\beta\ge 3$, we have $g\in C^3[0,1]$. A direct calculation gives
\[
M_3=\max_{0\le p \le 1}\Big|\,g''(p)+\tfrac{2p}{3}g'''(p)\,\Big|
= \tfrac{1}{3}\,\beta(\beta-1)(2\beta-1).
\]
Substituting this into the bound of Lemma~\ref{lem:interior} yields the claim.
\end{proof}

\begin{lemma}[Error between $u(t,p)$ and the reference solution $u^{\mathrm{ex}}(t)$]\label{lem:ex-error}
Define the pointwise error relative to the reference solution by
\[
\eta_i(t):=u_i^{\mathrm{ex}}(t)-u(t,p_i), \qquad 
\bm{\eta}(t)=(\eta_0(t),\dots,\eta_M(t))^\top, \qquad \eta_M(t)=0.
\]
Then, for all $0\le t\le T$,
\[
|\eta_i(t)|\;\le\; C(\theta)h^{3/2}.
\]
\end{lemma}

\begin{proof}
The condition $\eta_M(t)=0$ holds since we enforce $u^{\mathrm{ex}}_M(t)=u(t,1)$ at the boundary.  
Let $P:=I-\mathbf{e}_M\mathbf{e}_M^\top$ and $\widehat F_h:=P F_h P$. Then
\begin{align*}
\dot{\eta}_i(t)
&=\dot{u}_i^{\mathrm{ex}}(t)-\dot{u}(t,p_i) \\
&=-\theta\kappa(t)(F_h \bm{u}^{\mathrm{ex}})_i+\theta\kappa(t)\,(Fu)(t,p_i) \\
&=-\theta\kappa(t)(F_h \bm{\eta}(t))_i-\theta\kappa(t)\Big((F_h \bm{g})_i-(Fg)(p_i)\Big)y(t) \\
&=-\theta\kappa(t)(\widehat F_h \bm{\eta}(t))_i + r_i(t),
\end{align*}
where
\[
r_i(t) :=
\begin{cases}
-\theta\kappa(t)\Big((F_h \bm{g})_i-(Fg)(p_i)\Big)y(t), & 0\le i\le M-1, \\[6pt]
0, & i=M.
\end{cases}
\]

By Corollary~\ref{cor:g}, each entry of $\bm{r}(t)$ satisfies
\[
|r_i(t)|\le C(\theta)h^2|\dot y(t)|,
\]
hence
\[
\|\bm{r}(t)\|_2\le \sqrt{M}\,C(\theta)h^2|\dot y(t)|
= C(\theta)h^{3/2}|\dot y(t)|.
\]

Since $\widehat F_h^\dagger=-\widehat F_h$, it generates a unitary operator $U(t,s)$ on $\C^{M+1}$. Thus
\[
\bm{\eta}(t)=\int_0^t U(t,s)\,\bm{r}(s)\,ds,
\]
and because $U(t,s)$ is unitary,
\[
\|\bm{\eta}(t)\|_2\le \int_0^t \|\bm{r}(s)\|_2\,ds.
\]

Finally, using that $y$ is decreasing and $y(t)\in[0,1]$,
\[
\|\bm{\eta}(t)\|_2\le C(\theta)h^{3/2}\!\int_0^t |\dot y(s)|\,ds
= C(\theta)h^{3/2}(1-y(t))\le C(\theta)h^{3/2}.
\]
Therefore every entry satisfies $|\eta_i(t)|\le C(\theta)h^{3/2}$.
\end{proof}

\begin{lemma}[Finite propagation property of powers of $F_h$ {\cite[Lemma~3]{li2025linear}}]\label{lem:cone}
Let $F_h \in \C^{(M+1)\times (M+1)}$ denote the finite-difference discretization of 
$F = p\partial_p + \tfrac12$. Then, for every integer $k \ge 0$,
\[
\langle i | F_h^k | j \rangle = 0 \quad \text{whenever } |i-j| > k,
\]
and moreover
\[
|\langle i | F_h^k | j \rangle| \;\le\; h^{-k}.
\]
\end{lemma}

\begin{proof}
By construction, every nonzero entry of $F_h$ has magnitude at most $\tfrac{1}{2h}$. 
Moreover, $F_h$ can be decomposed as
\[
F_h = B_+ - B_-,
\]
where $B_+$ has nonzeros only on the first superdiagonal and $B_-$ has nonzeros only on the first subdiagonal. 
Thus $B_+$ moves support one index upward and $B_-$ one index downward. 

Hence, any product of $k$ such factors can move information by at most $k$ indices. 
This proves the property $\langle i|F_h^k|j\rangle=0$ whenever $|i-j|>k$. 

For magnitudes, note that each term of length $k$ in the binomial expansion
\[
F_h^k = (B_+ - B_-)^k
\]
has entries bounded by $(2h)^{-k}$. Since there are $2^k$ such terms in total, every entry of $F_h^k$ is bounded by $h^{-k}$.
\end{proof}

\begin{lemma}[Boundary mismatch error: $\bm{u}^{\mathrm{ex}}(t)$ vs.\ $\bm{u}^d(t)$]\label{lem:bdry-mismatch}
Let $\bm{\delta}(t):=\bm{u}^{\mathrm{ex}}(t)-\bm{u}^d(t)$ for $0\le t\le T$. 
If $\theta K_{\max}T \le \tfrac{1}{8e}$, then for $i\le \tfrac{3M}{4}$ we have
\begin{equation}
    |\delta_i(T)|\ \le\ \Big(2+\tfrac{M(1+C(\theta)h^{3/2})}{8e}\Big)\,2^{-M/4}.
    \label{eq:delta-final}
\end{equation}
\end{lemma}

\begin{proof}
Recall~\eqref{eq:ud-ode}, we can write the differential equation of boundary component of $\bm{u}^d(t)$
\[
\dot{u}^d_M(t)=-\kappa(t)\theta\Big(-\tfrac{2M-1}{4}\Big)u^d_{M-1}(t), 
\qquad u^d_M(0)=1.
\]
Write $F_h$ in block form as
\[
F_h=\begin{bmatrix} A & \bm{a}\\ -\bm{a}^\dagger & 0\end{bmatrix},
\qquad  \bm{a}=\begin{bmatrix}0&\cdots&0&a\end{bmatrix}^\top \in \R^M,\quad 
a=\tfrac{2M-1}{4}.
\]
Let $\mathcal{I}=\{0,\dots,M-1\}$ and denote by $\bm{u}^d_{\mathcal{I}}\in\C^M$ the subvector obtained by removing the last entry. Then
\[
\dot{\bm{u}}^d_{\mathcal{I}}=-\kappa(t)\theta\big(A \bm{u}^d_{\mathcal{I}}+u^d_M\,\bm{a}\big),
\qquad 
\dot{u}^d_M=-\kappa(t)\theta(-\bm{a}^\dagger)\bm{u}^d_{\mathcal{I}},
\]
while
\[
\dot{\bm{u}}^{\mathrm{ex}}_{\mathcal{I}}=-\kappa(t)\theta\big(A \bm{u}^{\mathrm{ex}}_{\mathcal{I}}+u^{\mathrm{ex}}_M\,\bm{a}\big), 
\qquad u^{\mathrm{ex}}_M(t)=y(t).
\]

We can decompose the error vector into its interior and boundary components as
\[
\bm{\delta}(t)=\begin{pmatrix}\bm{\delta}_{\mathcal{I}}(t)\\ \delta_M(t)\end{pmatrix}
=\begin{pmatrix}\bm{u}^{\mathrm{ex}}_{\mathcal{I}}(t)-\bm{u}^d_{\mathcal{I}}(t)\\[2pt] y(t)-u^d_M(t)\end{pmatrix}.
\]
Then $\bm{\delta}(t)$ satisfies the differential equation
\[
\dot{\bm{\delta}}(t)=-\kappa(t)\theta F_h \bm{\delta}(t)+\bm{b}(t),
\qquad \bm{\delta}(0)=\bm{0},
\]
where
\[
\bm{b}(t)=\Big(-\kappa(t)y(t)-\kappa(t)\theta\,\tfrac{2M-1}{4}\,u^{\mathrm{ex}}_{M-1}(t)\Big)\bm{e}_M
=:b(t)\,\bm{e}_M.
\]

Using Lemma~\ref{lem:ex-error}, $|u^{\mathrm{ex}}_{M-1}(t)-u(t,p_{M-1})|\le C(\theta)h^{3/2}$ and $|u(t,p_{M-1})|\le y(t)$, we obtain
\begin{align}
\Big|\!\int_0^T b(t)\,dt\Big|
&\le \Big|\!\int_0^T \dot{y}(t)\,dt\Big|
+\theta \tfrac{2M-1}{4}K_{\max}\!\int_0^T \big(y(t)+C(\theta)h^{3/2}\big)dt \notag\\
&\le (1-y(T))+\tfrac12\,\theta K_{\max}TM\,(1+C(\theta)h^{3/2}) \notag\\
&\le 1+\tfrac12\,\theta K_{\max}TM\,(1+C(\theta)h^{3/2}).
\label{eq:int-b-bound}
\end{align}

Since $F_h$ is skew-Hermitian, it generates the unitary operator
\[
U_h(t,s)=\exp\!\Big(-\int_s^t \kappa(\tau)\theta F_h\,d\tau\Big).
\]
Therefore, we can obtain
\[
\bm{\delta}(T)=\int_0^T U_h(T,s)\,\bm{b}(s)\,ds
=\int_0^T\sum_{m=0}^{\infty}\frac{1}{m!}\Big(-\theta\!\int_s^T\kappa(\tau)\,d\tau\Big)^m F_h^m\,b(s)\,\bm{e}_M\,ds.
\]

For $i\le \tfrac{3M}{4}$, Lemma~\ref{lem:cone} implies only terms with $m\ge M/4$ contribute. Thus
\begin{align*}
|\delta_i(T)|
&\le \Big|\!\int_0^T b(s)\,ds\Big| \sum_{m=\lceil M/4\rceil}^{\infty} \frac{(\theta K_{\max}T)^m}{m!\,h^m} \\
&\le \Big(1+\tfrac12\,\theta K_{\max}TM(1+C(\theta)h^{3/2})\Big)
\sum_{m=\lceil M/4\rceil}^{\infty} \Big(\frac{e\,\theta K_{\max}T}{m\,h}\Big)^m \\
&\le \Big(1+\tfrac12\,\theta K_{\max}TM(1+C(\theta)h^{3/2})\Big)
\sum_{m=\lceil M/4\rceil}^{\infty} (4e\,\theta K_{\max}T)^m.
\end{align*}
Using Stirling’s bound $m!\ge (m/e)^m$ in the second inequality above, together with $h=1/M$, we see that the last one becomes a geometric series. Under the assumption $\theta K_{\max}T \le \tfrac{1}{8e}$ so that $4e\theta K_{\max}T\le \tfrac12$, we obtain
\[
|\delta_i(T)|\ \le\ \Big(2+\tfrac{M(1+C(\theta)h^{3/2})}{8e}\Big)\,2^{-M/4}.
\]
\end{proof}

From \eqref{eq:triangle}, Lemma~\ref{lem:ex-error}, and Lemma~\ref{lem:bdry-mismatch},
\[
\big|u(T,p_i)-u_i^d(T)\big|
\ \le\
C(\theta)h^{3/2}
+\Big(2+\tfrac{M(1+C(\theta)h^{3/2})}{8e}\Big)2^{-M/4},
\qquad i\le \tfrac{3M}{4}.
\]
Finally, $(M/x)^\beta\le 4^\beta$ for $x\in\Imid$, proving \eqref{eq:final-topline}--\eqref{eq:main-bound}.
\hfill$\square$


\section{Quantum Circuit Implementation}
\label{sec:circuit}

\subsection{Preparation of the ancillary state \texorpdfstring{$|r_h\rangle$}{|r\_h}} 
\label{sec:prep-rh}

We now give an explicit construction of the normalized ancillary state
\begin{equation}
|r_h\rangle
= C_{M,\theta}^{-1}\sum_{i=0}^{M}\Big(\tfrac{i}{M}\Big)^{\beta}\,|i\rangle,
\qquad
\beta=\tfrac{1}{\theta}-\tfrac{1}{2},\qquad
C_{M,\theta}=\Big(\sum_{i=0}^{M} (i/M)^{2\beta}\Big)^{1/2},
\end{equation}
where $M+1=2^m$ for some $m\in\mathbb{N}$. Our method utilizes the combination of LCU and QSVT. We first implement a block-encoding of the diagonal operator
\[
\hat H_{\mathrm{init}}:=\sum_{i=0}^{M}\tfrac{i}{M}\,|i\rangle\!\langle i|,
\]
and then apply a QSVT sequence to perform the monomial singular value transformation $f(x)=x^\beta$. Finally, postselection on the ancilla registers produces $|r_h\rangle$.

Using the binary representation $i=\sum_{k=0}^{m-1}2^k b_k$ and the relation $Z_k|b_k\rangle=(-1)^{b_k}|b_k\rangle$, the operator can be written in Pauli operators as
\begin{equation}
\hat H_{\mathrm{init}}=\underbrace{\frac{1}{2}}_{\omega_0}I-\sum_{k=0}^{m-1}\underbrace{\frac{2^k}{2M}}_{\omega_{k+1}}Z_k.
\label{eq:H_init}
\end{equation}
Let $a:=\lceil \log_2(m+1)\rceil$ and define nonnegative weights $\omega_0=\tfrac{1}{2}$ and $\omega_{k+1}=\frac{2^k}{2M}$ for $0\le k\le m-1$, so that $\sum_{j=0}^m \omega_j=1$. With the $a$-qubit ancilla register we prepare 
\(
\mathrm{Prep}_\mathrm{init}\,|0\rangle^{\otimes a}=\sum_{j=0}^m \sqrt{\omega_j}\,|j\rangle
\)
and apply
\(
\mathrm{Select}_\mathrm{init}=\sum_{j=0}^m |j\rangle\!\langle j|\otimes U_j
\)
with $U_0=I$ and $U_{k+1}=-Z_k$ for $0 \le k \le m-1$. The unitary operator
\(
U_{\mathrm{init}}=(\mathrm{Prep}_\mathrm{init}^\dagger\otimes I)\,\mathrm{Select}_\mathrm{init}\,(\mathrm{Prep}_\mathrm{init}\otimes I)
\)
is then a $(1,a,0)$ block encoding of $\hat H_{\mathrm{init}}$. Since all amplitudes $\sqrt{\omega_j}$ are real and nonnegative, $\mathrm{Prep}_{\mathrm{init}}$ uses only $R_Y$ rotations; with the M\"ott\"onen scheme it requires exactly $(2^a-2)$ CNOTs and $(2^a-1)$ $R_Y$ rotations, which scales as $O(\log M)$ since $2^{a-1}<m+1\le 2^a$. The explicit decomposition for $a=3$ is shown in Fig.~\ref{fig:prep-mottonen}.

\begin{figure}[t]
\centering
\scalebox{0.73}{
    \begin{quantikz}
        & \gate[3]{\mathrm{Prep}} & \midstick[3,brackets=none]{=} & \gate{R_Y(\theta_1)} \slice[style=black]{} & & \ctrl{1} & & \ctrl{1} \slice[style=black]{} &&&& \ctrl{2} &&&& \ctrl{2} & \\
        &                         &                                &                         & \gate{R_Y(\theta_2)} & \targ{} & \gate{R_Y(\theta_3)} & \targ{} && \ctrl{1} &&&& \ctrl{1} && & \\
        &                         &                                &                         &                     &        &                      &         & \gate{R_Y(\theta_4)} & \targ{} & \gate{R_Y(\theta_5)} & \targ{} & \gate{R_Y(\theta_6)} & \targ{} & \gate{R_Y(\theta_7)} & \targ{} &
    \end{quantikz}
}
\caption{\textbf{Decomposition of $\mathrm{Prep}_{\mathrm{init}}$ for $a=3$.}  
Since all amplitudes are real and nonnegative, $\mathrm{Prep}_{\mathrm{init}}$ uses only $R_Y$ rotations.  
Gate counts: $2^a-2=6$ CNOT gates and $2^a-1=7$ $R_Y$ gates.}
\label{fig:prep-mottonen}
\end{figure}
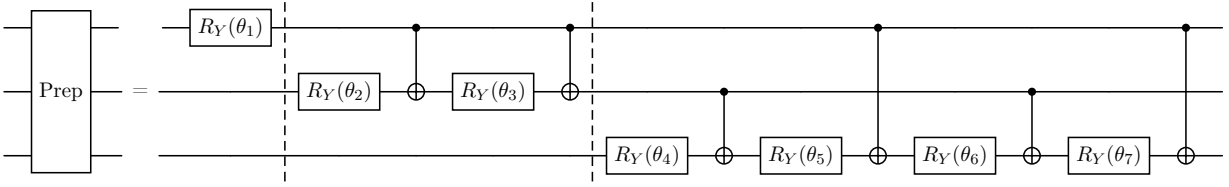

The $\mathrm{Select}_\mathrm{init}$ unitary is a bank of multi-controlled operations conditioned on the $a$-qubit ancilla register; in particular, there are $m$ instances of $a$-bit controlled $U_{k+1}=-Z_k$ gates, with controls on the $a$-qubit ancilla register and targets on the register where $\ket{r_h}$ is implemented. According to~\cite{khattar2025rise}, each $a$-bit Toffoli can be decomposed into $(2a-3)$ standard Toffoli gates with one additional ancilla qubit, so the overall Toffoli count for $\mathrm{Select}_\mathrm{init}$ is $m(2a-3)$. The full block encoding structure is shown in Fig.~\ref{fig:lcu-hinit}.

\begin{figure}[t]
\centering
\scalebox{0.7}{
    \begin{quantikz}
        \lstick[3]{$\ket{0}^{\otimes a}$} & \gate[3]{\mathrm{Prep}_\mathrm{init}} & \gategroup[10,steps=10,style={dashed,rounded
        corners,fill=blue!12, inner xsep=2pt},background,label style={label position=above,anchor=north,yshift=-0.2cm}]{\(\mathrm{Select}_\mathrm{init}=\sum_j |j\rangle\!\langle j|\otimes U_j\)} & \ctrl[open]{2} & & \ctrl{3}&\ctrl[open]{4}&\ctrl{5}&\ctrl[open]{6}&\ctrl{7}&\ctrl[open]{8}&\ctrl{9}&\gate[3]{\mathrm{Prep}_\mathrm{init}^\dagger} &\\
        & & & \ocontrol{} & & \ocontrol{}&\control{}&\control{}&\ocontrol{}&\ocontrol{}&\control{}&\control{} &&\\
        & & \gate{X} & \gate{Z} & \gate{X} & \ocontrol{}&\ocontrol{}&\ocontrol{}&\control{}&\control{}&\control{}&\control{} &&\\
        \lstick[7]{$\ket{+}^{\otimes m}$}&&&&& \gate{Z} & & & & & & &&\\
        &&&&& & \gate{Z} & & & & & &&\\
        &&&&& & & \gate{Z} & & & & &&\\
        &&&&& & & & \gate{Z} & & & &&\\
        &&&&& & & & & \gate{Z} & & &&\\
        &&&&& & & & & & \gate{Z} & &&\\
        &&&&& & & & & & & \gate{Z} &&
    \end{quantikz}
}
\caption{\textbf{Block encoding of $\hat H_{\mathrm{init}}$ using LCU.}  
$\mathrm{Prep}_\mathrm{init}$ makes $\sum_j\sqrt{\omega_j}|j\rangle$, while $\mathrm{Select}_\mathrm{init}$ applies controlled $\{I,-Z_0,-Z_1,\dots,-Z_{m-1}\}$ up to global phase. The construction $(\mathrm{Prep}_\mathrm{init}^\dagger\!\otimes I)\,\mathrm{Select}_\mathrm{init}\,(\mathrm{Prep}_\mathrm{init}\otimes I)$ realizes a $(1,a,0)$ block-encoding of $\hat H_{\mathrm{init}}$.}
\label{fig:lcu-hinit}
\end{figure}
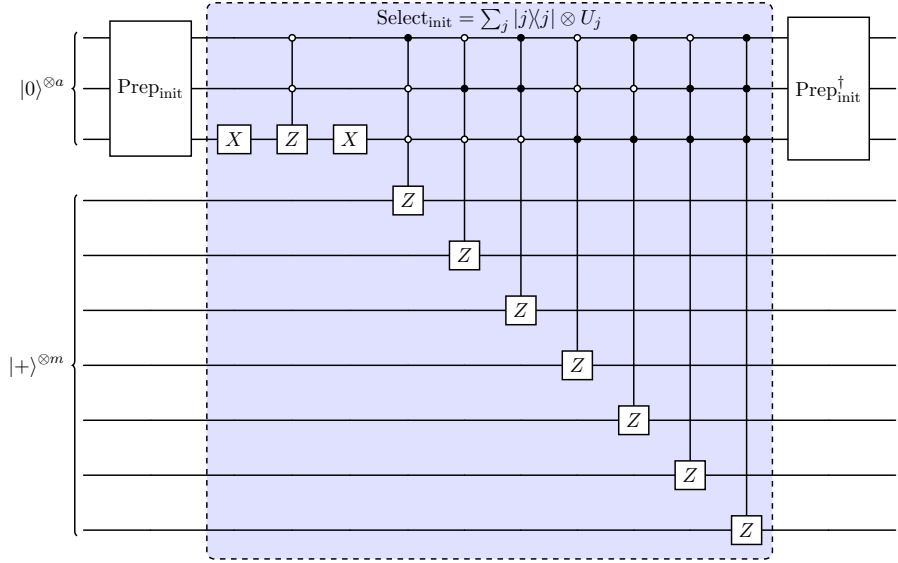

To encode the monomial factor $(i/M)^\beta$ into the computational basis, we take $U_{\mathrm{init}}$ as the signal unitary for QSVT and implement the monomial $f(x)=x^\beta$ of degree $\beta$. The sequence consists of exactly $\beta$ alternating applications of $U_{\mathrm{init}}$ and $U_{\mathrm{init}}^\dagger$, interleaved with single-qubit $R_Z(2\phi_j)$ rotations on the signal qubit, yielding
\[
(\langle 0|^{\otimes a}\otimes I)\,U_{\mathrm{init},\Phi}\,(|0\rangle^{\otimes a}\otimes I)
=\sum_{i=0}^{M}\Big(\tfrac{i}{M}\Big)^{\beta}|i\rangle\!\langle i|.
\]
The structure of the QSVT sequence appears in Fig.~\ref{fig:qsvt-rh}. Here $\Pi=|0\rangle\!\langle 0|^{\otimes a}$, and in the diagram the $C_{\Pi}\mathrm{NOT}$ gate denotes the controlled operation $X\otimes\Pi + I\otimes(I-\Pi)$. Each $C_{\Pi}\mathrm{NOT}$ gate can be realized using an $a$-bit Toffoli gate and $X$ gates. The QSVT construction requires $a$ ancilla qubits for block encoding, plus one extra ancilla qubit for the signal processing register. Thus, the total number of ancilla qubits is $a+1$. If the $a$-bit Toffoli gates are decomposed into standard Toffolis, an additional ancilla qubit is required. Successful postselection requires that all of them be measured in $|0\rangle$ simultaneously in order to obtain the desired $|r_h\rangle$ state on the $m$-qubit register. The QSVT sequence also entails $2\beta$ $a$-bit Toffoli gates and $\beta$ single-qubit $R_Z$ gates.

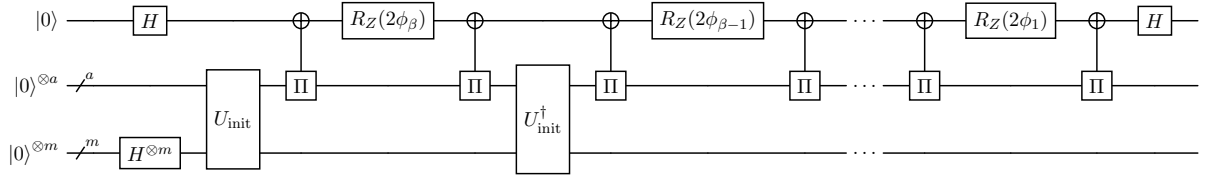
\begin{figure}[t]
\centering
\scalebox{0.7}{
    \begin{quantikz}
        \lstick[1]{\ket{0}} & & \gate{H} & & \targ{}\wire[d][1]{q} & \gate{R_Z(2\phi_\beta)} & \targ{}\wire[d][1]{q} & & \targ{}\wire[d][1]{q} & \gate{R_Z(2\phi_{\beta-1})} & \targ{}\wire[d][1]{q} & \ \ldots\ & \targ{}\wire[d][1]{q} & \gate{R_Z(2\phi_1)} & \targ{}\wire[d][1]{q} & \gate{H} & \\
        \lstick[1]{$\ket{0}^{\otimes a}$} & \qwbundle{a} & & \gate[2]{U_{\mathrm{init}}} & \gate{\Pi}   & & \gate{\Pi} & \gate[2]{U_{\mathrm{init}}^\dagger} & \gate{\Pi} & & \gate{\Pi} & \ \ldots\ & \gate{\Pi} & & \gate{\Pi} &&\\
        \lstick[1]{$\ket{0}^{\otimes m}$} & \qwbundle{m} & \gate{H^{\otimes m}} & & & & & & & & & \ \ldots\ & &&&&\\
    \end{quantikz}
}
\caption{\textbf{QSVT sequence for the monomial $f(x)=x^\beta$.}  
Here $\Pi=|0\rangle\!\langle 0|^{\otimes a}$. In the diagram, the $C_{\Pi}\mathrm{NOT}$ gate denotes the controlled operation $X\otimes\Pi + I\otimes(I-\Pi)$. The interleaving phases $\{\phi_j\}$ implement the degree-$\beta$ monomial transform.}
\label{fig:qsvt-rh}
\end{figure}

Applied to the uniform superposition $\ket{+}^{\otimes m}$, the overall circuit produces
\[
U_{\mathrm{init},\Phi}(|0\rangle^{\otimes a}\otimes |+\rangle^{\otimes m})
=|0\rangle^{\otimes a}\frac{1}{\sqrt{M+1}}\sum_{i=0}^M \Big(\tfrac{i}{M}\Big)^\beta |i\rangle+\ket{\perp},
\]
where $U_{\mathrm{init},\Phi}$ denotes the QSVT sequence. Thus, postselecting all ancilla qubits in the state $|0\rangle$ yields $|r_h\rangle$ with success probability
\[
P_{\mathrm{succ}}=\frac{1}{M+1}\sum_{i=0}^M\Big(\tfrac{i}{M}\Big)^{2\beta}
\ge \frac{M}{M+1}\cdot\frac{1}{2\beta+1}.
\]
Hence, the expected number of trials is $O(\beta)$, which can be reduced to $O(\sqrt{\beta})$ using amplitude amplification.

\begin{algorithm}[t]
\caption{Preparation of the ancillary state $|r_h\rangle$ using LCU and QSVT}
\label{alg:prep-rh}
\begin{algorithmic}[1]
\State \textbf{Inputs:} $M+1=2^m$, $\beta\in\mathbb{N}$, $a=\lceil\log_2(m+1)\rceil$.
\State \textbf{Initialize} the $a$-qubit ancilla register in $|0\rangle^{\otimes a}$ and the data register in $|+\rangle^{\otimes m}$.
\State \textbf{Block-encode} $\hat H_{\mathrm{init}}=\sum_i (i/M)|i\rangle\!\langle i|$ using $\mathrm{Prep}_\mathrm{init}$ and $\mathrm{Select}_\mathrm{init}$:
\Statex \hspace{1.2em} (i) Prepare $\sum_j\sqrt{\omega_j}|j\rangle$ on $a$-qubit register with $\mathrm{Prep}_\mathrm{init}$. The weights $\{\omega_j\}_{j=0}^m$ are given in Eq.~\eqref{eq:H_init}.
\Statex \hspace{1.2em} (ii) Apply $\mathrm{Select}_\mathrm{init}=\sum_{j=0}^m |j\rangle\!\langle j|\otimes U_j$ where $U_0=I$ and $U_{k+1}=-Z_k$ for $0\le k \le m-1$.
\Statex \hspace{1.2em} (iii) Unprepare with $\mathrm{Prep}_\mathrm{init}^\dagger$ to obtain $U_{\mathrm{init}}$.
\State \textbf{Apply QSVT:} Add one extra ancilla as the \emph{signal processing register}, and apply a degree-$\beta$ sequence alternating $U_{\mathrm{init}}$ and $U_{\mathrm{init}}^\dagger$ with phases $\{\phi_j\}_{j=1}^\beta$ to realize a block encoding of $\sum_i (i/M)^\beta |i\rangle\!\langle i|$.
\State \textbf{Postselect} all $(a+1)$ ancillas in $|0\rangle$ to obtain $|r_h\rangle$ on the $m$-qubit register. We can use amplitude amplification to boost the postselection success probability.
\State \textbf{Resources:} $\mathrm{Prep}_\mathrm{init}$ uses $(2^a-2)$ CNOTs and $(2^a-1)$ $R_Y$; 
$\mathrm{Select}_\mathrm{init}$ uses $m$ $a$-bit Toffolis; 
QSVT uses $\beta$ calls to $U_{\mathrm{init}}^{(\dagger)}$, $2\beta$ $a$-bit Toffolis for reflections, and $\beta$ single-qubit $R_Z$ gates, with one extra signal-processing ancilla.  
\State \textbf{Complexity:} Accounting for the decomposition of $a$-bit Toffoli gates into standard Toffolis, the per-attempt cost is $O(\beta\log M\cdot\log\log M)$. With postselection success probability $O(1/\beta)$, the total cost is $O(\beta^2\log M\cdot\log\log M)$ (or $O(\beta^{3/2}\log M\cdot\log\log M)$ using amplitude amplification).
\end{algorithmic}
\end{algorithm}

\begin{table}[t]
\centering
\caption{\textbf{Example of QSVT phase angles for the monomial $f(x)=x^\beta$.}
The values $\{\phi_j\}_{j=1}^\beta$ (in radians) correspond to the single-qubit rotation phases used in the QSVT circuit of Fig.~\ref{fig:qsvt-rh}.}
\label{tab:qsvt-angles-monomial}
\begin{tabular}{@{}c l@{}}
\toprule
$\beta$ & $\{\phi_j\}_{j=1}^{\beta}$ (ascending order) \\
\midrule
3 &
$\{-1.945530537814129,\ -2.1688268601597227,\ -2.1688268601597227\}$ \\
4 &
$\{-0.17915969502442763,\ -1.9634951462137356,\ -2.1770342706081474,$\\
& $\ -1.9634951462137356\}$ \\
5 &
$\{1.4843149138525842,\ -1.8078352881528696,\ -2.0759142978060185,$ \\
& $\ -2.0759142978060185,\ -1.8078352881528696\}$ \\
6 &
$\{3.099514146455192,\ -1.7077184397821685,\ -1.9424558926637125,$ \\
& $\ -2.0823497396925856,\ -1.9424558926637125,\ -1.7077184397821685\}$ \\
7 &
$\{-1.5913870208780079,\ -1.648016853210964,\ -1.8228649318945727,$ \\
& $\ -2.0166094870933566,\ -2.0166094870933566,\ -1.8228649318945727,\ -1.648016853210964\}$ \\
\bottomrule
\end{tabular}
\vspace{0.5em}

\begin{minipage}{0.95\linewidth}
\textit{Note.} 
The phases were obtained by numerically solving with QSPPACK (MATLAB)~\cite{dong2021efficient}, with parameters chosen such that the synthesized polynomial \(p_\Phi(x)\) satisfies
\[
\sup_{x\in[-1,1]} |p_\Phi(x)-x^\beta| \le 10^{-12}.
\]
Since QSPPACK outputs phases for QSP, an additional postprocessing step was applied to adapt them to the structure required by QSVT.
\end{minipage}
\end{table}

\subsection{Block-encoding of \texorpdfstring{$\theta F_h$}{theta Fh} and \texorpdfstring{$I\otimes H + i\theta F_h\otimes K$}{IH + i theta Fh K}}
\label{sec:block-encode-Fh-and-assemble}

We first construct a block encoding of the ancillary operator $\theta F_h$, following the approach of Sec.~\ref{sec:prep-rh}. Assuming that block encodings of $H$ and $K$ are available, we then assemble these ingredients to obtain a block encoding of the operator $I\otimes H + i\theta F_h\otimes K$.

Recall the tridiagonal skew-Hermitian stencil $F_h$ from~\eqref{eq:Fh}. Using the diagonal operator
\begin{equation}
D=\theta\,(2M\,\hat H_{\mathrm{init}}+I)=\theta\cdot\mathrm{diag}(1,3,5,\dots,2M+1),
\qquad
\hat H_{\mathrm{init}}=\sum_{i=0}^M \frac{i}{M}\,|i\rangle\!\langle i|,
\end{equation}
and the right-shift operator $R=\sum_{i=0}^{M-1}|i\rangle\!\langle i+1|$, one can verify
\begin{equation}
\theta F_h=\frac{1}{4}\big(DR-R^\dagger D\big).
\label{eq:Fh-comm}
\end{equation}
Hence it suffices to block-encode $D$ and $R$ and then combine them.

\paragraph*{Block encoding of $D$.}
Let $m=\log_2(M+1)$ and $a=\lceil\log_2(m+1)\rceil$ so that $2^{a-1}<m+1\le 2^a$. Write
\begin{equation}
D \;=\; \theta(M\!+\!1)\,I - \sum_{k=0}^{m-1} \theta\,2^k\,Z_k
\;=\; \alpha_D\!\Bigg(\underbrace{\frac{\theta(M+1)}{\alpha_D}}_{\omega'_0} I
- \sum_{k=0}^{m-1} \underbrace{\frac{\theta 2^k}{\alpha_D}}_{\omega'_{k+1}} Z_k\Bigg),
\qquad
\alpha_D=\theta(2M+1).
\label{eq:D}
\end{equation}
Define the unitaries $U_0=I$ and $U_{k+1}=-Z_k$ for $0\le k\le m-1$, and let \(\{\omega'_j\}_{j=0}^m\) be the nonnegative weights defined above, which sum to~1. Prepare the $a$-qubit register with $\mathrm{Prep}_D\,|0\rangle^{\otimes a}=\sum_{j=0}^m \sqrt{\omega'_j}\,|j\rangle,$ and apply $\mathrm{Select}_D=\sum_{j=0}^m |j\rangle\!\langle j|\otimes U_j.$ Then the unitary operator
\begin{equation}
U_D=(\mathrm{Prep}_D^\dagger\!\otimes I)\,\mathrm{Select}_D\,(\mathrm{Prep}_D\!\otimes I)
=\begin{bmatrix} D/\alpha_D & * \\ * & * \end{bmatrix}
\end{equation}
is an $(\alpha_D,a,0)$ block encoding of $D$. As in Sec.~\ref{sec:prep-rh}, all amplitudes are real and nonnegative, so $\mathrm{Prep}_D$ uses only $R_Y$ rotations and can be realized by the Möttönen scheme with exactly $(2^a-2)$ CNOTs and $(2^a-1)$ $R_Y$ gates, which scales as $O(\log M)$ since $a=O(\log\log M)$. The $\mathrm{Select}_D$ comprises $m$ instances of $a$-bit controlled $U_{k+1}=-Z_k$ gates. As mentioned in Sec.~\ref{sec:prep-rh}, each $a$-bit Toffoli decomposes into $(2a-3)$ Toffolis with one additional ancilla qubit, hence the Toffoli count is $m(2a-3)$~\cite{khattar2025rise}.

\paragraph*{Block encoding of $R$ using QFT-adder.}
We can implement the right shift operator $R$ via a QFT-adder~\cite{draper2000addition}. Let $N=2^{m+1}$ and consider the $(m\!+\!1)$-qubit QFT $F_N$. Define
\begin{equation}
U_R \;=\; F_N^\dagger\,\mathrm{diag}(\omega^0,\omega^{-1},\dots,\omega^{-(N-1)})\,F_N,
\qquad
\omega=e^{2\pi i/N}.
\end{equation}
Since 
\(
\mathrm{diag}(\omega^0,\omega^{-1},\dots,\omega^{-(N-1)})
= \bigotimes_{k=0}^{m} R_Z(-\pi/2^{k}),
\)
this diagonal operator can be implemented with $(m\!+\!1)$ parallel single-qubit $R_Z$ gates. A textbook QFT on $m\!+\!1$ qubits requires $\frac{m(m+1)}{2}$ controlled-phase gates together with $(m\!+\!1)$ Hadamard gates. Acting on $\mathbb{C}^{2^{m+1}}$, $U_R$ implements the cyclic shift $|i\rangle\mapsto |(i+1)\bmod 2^{m+1}\rangle$. When restricted to the $m$-qubit subspace, $U_R$ realizes a $(1,1,0)$ block encoding of
\(
R=\sum_{i=0}^{M-1} |i\rangle\!\langle i+1|,
\)
as shown in Fig.~\ref{fig:block-encoding-R}. In total, the $R$ block encoding requires one ancilla qubit, $O((\log M)^2)$ two-qubit gates, and $O(\log M)$ single-qubit gates.

\begin{figure}[t]
    \centering
    \begin{subfigure}{0.4\textwidth}
        \centering
        \scalebox{0.7}{
            \begin{quantikz}
            \lstick[4]{$m$ qubits}& & \gate[5]{\mathrm{QFT}} & \gate{R_Z\left(-\frac{\pi}{2^m}\right)} & \gate[5]{\mathrm{QFT}^\dagger}&&\\
            \setwiretype{n} & \vdots & & \vdots & & \vdots &\\
            & & & \gate{R_Z\left(-\frac{\pi}{4}\right)} & &&\\
            & & & \gate{R_Z\left(-\frac{\pi}{2}\right)} & &&\\
            \lstick[1]{ancilla \ket{0}}& & & \gate{R_Z\left(-\pi\right)} & &&
            \end{quantikz}
        }
        \caption{}
        \label{fig:block-encoding-R}
    \end{subfigure}
    \hfill
    \begin{subfigure}{0.56\textwidth}
        \centering
        \scalebox{0.7}{
            \begin{quantikz}
            & \qwbundle{m} &&& \gate[2]{U_R} & \gate[3]{U_D} & & & \gate[3]{U_D} & \gate[2]{U^\dagger_R} & &\\
            \lstick[5]{$\ket{0}^{\otimes(2a+3)}$}& \qwbundle{1} &&& & & \swap{2} & & & & &\\
            & \qwbundle{a} &&& & & & \swap{2} & & & &\\
            & \qwbundle{1} &\ghost{Z}&& & & \targX{} & & & & &\\
            & \qwbundle{a} &\ghost{Z}&& & & & \targX{} &  & & &\\
            & \qwbundle{1} & \gate{H} & \gate{Z} & \ctrl[open]{-5} & \ctrl[open]{-4} & & & \ctrl{-4} & \ctrl{-5} & \gate{H} &\\
            \end{quantikz}
        }
        \caption{}
        \label{fig:block-encoding-Fh}
    \end{subfigure}
    \caption{\textbf{Block encoding of $R$ via a QFT adder and block encoding of $\theta F_h$.} 
    (a) The $(m\!+\!1)$-qubit QFT layer $F_{2^{m+1}}$ together with the parallel phase gates $\bigotimes_{k=0}^{m} R_Z(-\pi/2^{k})$ realizes the cyclic shift. This structure gives a $(1,1,0)$ block encoding of $R=\sum_{i=0}^{M-1}|i\rangle\!\langle i+1|$. 
    (b) The LCU structure produces $\tfrac{1}{2}(U_DU_R-U_R^\dagger U_D)$ on the computational block, yielding an $(\alpha_{\theta F}, a_{\theta F}, 0)$ block encoding of $\theta F_h$ with $\alpha_{\theta F}=\alpha_D/2=\tfrac{\theta}{2}(2M+1)$ and $a_{\theta F}=2a+3=O(\log\log M)$.}
\end{figure}
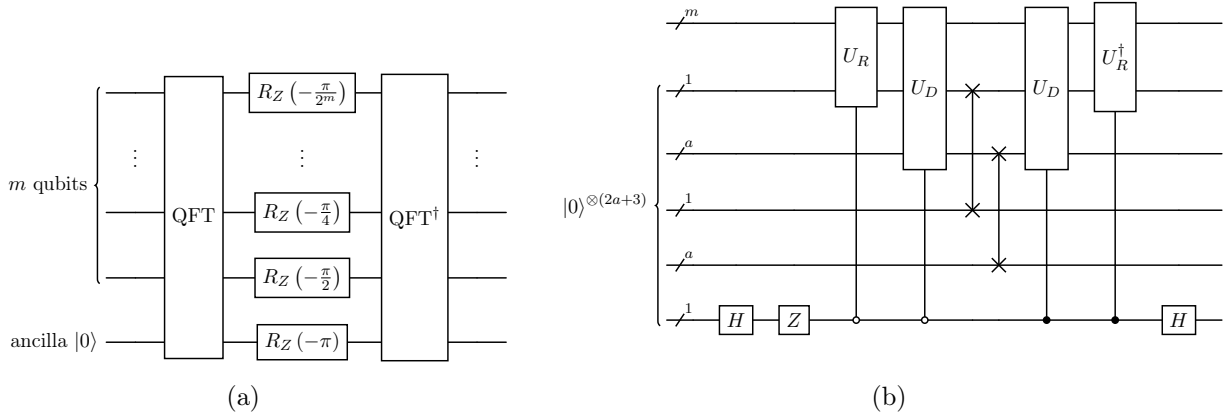

\paragraph*{Combining $U_D$ and $U_R$.}
We now combine $U_D$ and $U_R$ to realize $\theta F_h$ as given in~\eqref{eq:Fh-comm}. Consider the Hadamard gate on a single control qubit that prepares the equal superposition of two branches and implements
\begin{equation}
\frac{1}{2}\Big(U_D U_R \;-\; U_R^\dagger U_D\Big),
\end{equation}
with relative signs enforced by inserting a $Z$ gate on the control as shown in Fig.~\ref{fig:block-encoding-Fh}. Since $U_D$ block-encodes $D/\alpha_D$ and $U_R$ block-encodes $R$, the top-left block of the above unitary equals
\[
\frac{1}{2}\Big(DR/\alpha_D \;-\; R^\dagger D/\alpha_D\Big)
= \frac{1}{2\alpha_D}\,(DR-R^\dagger D).
\]
Comparing with $\theta F_h=\frac{1}{4}(DR-R^\dagger D)$ shows that the resulting block-encoding has normalization
\begin{equation}
\alpha_{\theta F}=\frac{\alpha_D}{2}=\frac{\theta}{2}(2M+1),
\qquad
a_{\theta F}=2a+3=O(\log\log M).
\end{equation}
The controlled versions of $U_D$ and $U_R$ required by the LCU branching can be obtained by adding one more control line to their primitive gates such as single-qubit rotations, CNOTs, Toffolis, and controlled phase operations in QFT. Using the multi-controlled gate decompositions of~\cite{khattar2025rise}, the resulting controlled circuits incur only a constant-factor overhead in depth compared to their uncontrolled versions.

\begin{algorithm}[t]
\caption{Block encoding of $\theta F_h$ using $D$ and $R$}
\label{alg:block-encode-thetaFh}
\begin{algorithmic}[1]
\State \textbf{Inputs:} $M+1=2^m$, $a=\lceil\log_2(m+1)\rceil$.
\State \textbf{Implementation of $U_D$:} Prepare $\sum_{j=0}^m \sqrt{\omega'_j}\,|j\rangle$ on an $a$-qubit register (\(\mathrm{Prep}_D\)) with \(\{\omega'_j \}_{j=0}^m\) in~\eqref{eq:D}; apply \(\mathrm{Select}_D=\sum_j |j\rangle\!\langle j|\otimes U_j\) with $U_0=I$, $U_{k+1}=-Z_k$; unprepare with $\mathrm{Prep}_D^\dagger$ to obtain a $(\alpha_D,a,0)$ block-encoding of $D$.
\State \textbf{Implementation of $U_R$:} Implement the cyclic shift $U_R=F_{2^{m+1}}^\dagger \big(\bigotimes_{k=0}^{m} R_Z(-\pi/2^k)\big) F_{2^{m+1}}$, which is a $(1,1,0)$ block-encoding of $R$.
\State \textbf{Linear combination:} Use the LCU method to realize $\displaystyle \frac{1}{2}\big(U_DU_R - U_R^\dagger U_D\big) \;\propto\; \frac{1}{4}\big(DR - R^\dagger D\big)$ on the system, giving a $(\alpha_{\theta F},\,2a{+}3,\,0)$ block-encoding of $\theta F_h$ with $\alpha_{\theta F}=\frac{\alpha_D}{2}=\frac{\theta}{2}(2M{+}1)$.
\State \textbf{Output:} $U_{\theta F}$, a $(\alpha_{\theta F},2a{+}3,0)$ block-encoding of $\theta F_h$.
\State \textbf{Resources:}
\Statex \ \ Ancillas: $2a+3=O(\log\log M)$ in total (When the decomposition of multi-controlled gates is taken into account, a few additional ancilla qubits are required).
\Statex \ \ Gate complexity: two uses of $U_D$ and two uses of $U_R$.
\Statex \ \ \hspace{1.75em} $U_D$: $O(\log M\cdot\log\log M)$ gates; \quad $U_R$: $O((\log M)^2)$ gates.
\Statex \ \ Overall: $O\big((\log M)^2\big)$ gates to realize $U_{\theta F}$.
\end{algorithmic}
\end{algorithm}

We now assemble the total operator, assuming block encodings of $H$, $K$, and $\theta F_h$ are available. Let $U_H$ denote an $(\alpha_H,a_H,\epsilon_H)$ block encoding of $H$, $U_K$ denote an $(\alpha_K,a_K,\epsilon_K)$ block encoding of $K$, and $U_{\theta F}$ denote an $(\alpha_{\theta F},a_{\theta F},\epsilon_{\theta F})$ block encoding of $\theta F_h$.

That is,
\[
(\langle 0|^{\otimes a_H}\!\otimes I)\,U_H\,(|0\rangle^{\otimes a_H}\!\otimes I)=\tfrac{H}{\alpha_H}+E_H,\quad
(\langle 0|^{\otimes a_K}\!\otimes I)\,U_K\,(|0\rangle^{\otimes a_K}\!\otimes I)=\tfrac{K}{\alpha_K}+E_K,
\]
\[
(\langle 0|^{\otimes a_{\theta F}}\!\otimes I)\,U_{\theta F}\,(|0\rangle^{\otimes a_{\theta F}}\!\otimes I)=\tfrac{\theta F_h}{\alpha_{\theta F}}+E_{\theta F},
\quad
\|E_H\|\le\epsilon_H,\ \|E_K\|\le\epsilon_K,\ \|E_{\theta F}\|\le\epsilon_{\theta F}.
\]

Consider \(U_{\theta F\otimes K}:=U_{\theta F}\otimes U_K\). A direct calculation shows that
\[
\big((\langle 0|^{\otimes a_{\theta F}}\!\otimes I)\otimes(\langle 0|^{\otimes a_K}\!\otimes I)\big)\,U_{\theta F\otimes K}\,\big((|0\rangle^{\otimes a_{\theta F}}\!\otimes I)\otimes(|0\rangle^{\otimes a_K}\!\otimes I)\big)
=\frac{\theta F_h\otimes K}{\alpha_{\theta F}\alpha_K}+E_{\theta F\otimes K},
\]
where
\[
E_{\theta F\otimes K}
:= \tfrac{\theta F_h}{\alpha_{\theta F}}\!\otimes E_K\;+\;E_{\theta F}\!\otimes\tfrac{K}{\alpha_K}\;+\;E_{\theta F}\!\otimes E_K,
\qquad
\therefore \; \|E_{\theta F\otimes K}\|\le \epsilon_{\theta F}+\epsilon_K+\epsilon_{\theta F}\epsilon_K=:\epsilon_{\theta F\otimes K}.
\]
Thus, \(U_{\theta F\otimes K}\) is an \((\alpha_{\theta F}\alpha_K,\ a_{\theta F}+a_K,\ \epsilon_{\theta F\otimes K})\) block encoding of \(\theta F_h\otimes K\).

Now we can use the LCU method to build $I\!\otimes\!H + i\,\theta F_h\!\otimes\!K$. Introduce a single-qubit operator
\[
\mathrm{Prep}_{\mathrm{tot}}
=R_Y\!\left(2\arctan\sqrt{\tfrac{\alpha_{\theta F}\alpha_K}{\alpha_H}}\right),
\quad
\mathrm{Prep}_{\mathrm{tot}}|0\rangle
=\frac{\sqrt{\alpha_H}\,|0\rangle+\sqrt{\alpha_{\theta F}\alpha_K}\,|1\rangle}{\sqrt{\alpha_H+\alpha_{\theta F}\alpha_K}},
\]
and define the select unitary
\[
\mathrm{Select}_{\mathrm{tot}}
=|0\rangle\!\langle 0|\otimes (I\!\otimes\!U_H)\;+\;|1\rangle\!\langle 1|\otimes \big(i\,U_{\theta F\otimes K}\big),
\]
where the phase factor \(i\) is implemented by an \(S\) gate to the control qubit.
Set
\[
U_{\mathrm{tot}}:=(\mathrm{Prep}_{\mathrm{tot}}^\dagger\!\otimes I)\ \mathrm{Select}_{\mathrm{tot}}\ (\mathrm{Prep}_{\mathrm{tot}}\!\otimes I).
\]
Projecting the joint ancilla onto \(|0\rangle^{\otimes(1+a_H+a_{\theta F}+a_K)}\) yields the block
\[
\frac{\alpha_H}{\alpha_H+\alpha_{\theta F}\alpha_K}\left(\tfrac{I\otimes H}{\alpha_H}+E_H\right)
\;+\;
\frac{\alpha_{\theta F}\alpha_K}{\alpha_H+\alpha_{\theta F}\alpha_K}\left(\tfrac{i\,\theta F_h\otimes K}{\alpha_{\theta F}\alpha_K}+i\,E_{\theta F\otimes K}\right),
\]
which equals
\[
\frac{I\otimes H+i\,\theta F_h\otimes K}{\alpha_H+\alpha_{\theta F}\alpha_K}
\;+\;
E_{\mathrm{tot}},
\qquad
\|E_{\mathrm{tot}}\|
\le
\frac{\alpha_H}{\alpha_H+\alpha_{\theta F}\alpha_K}\,\epsilon_H
+
\frac{\alpha_{\theta F}\alpha_K}{\alpha_H+\alpha_{\theta F}\alpha_K}\,\epsilon_{\theta F\otimes K}=:\epsilon_{\mathrm{tot}}.
\]
Therefore, \(U_{\mathrm{tot}}\) is an \(\big(\alpha_H+\alpha_{\theta F}\alpha_K,\ 1+a_H+a_{\theta F}+a_K,\ \epsilon_{\mathrm{tot}}\big)\) block encoding of \(I\!\otimes\!H+i\,\theta F_h\!\otimes\!K\).

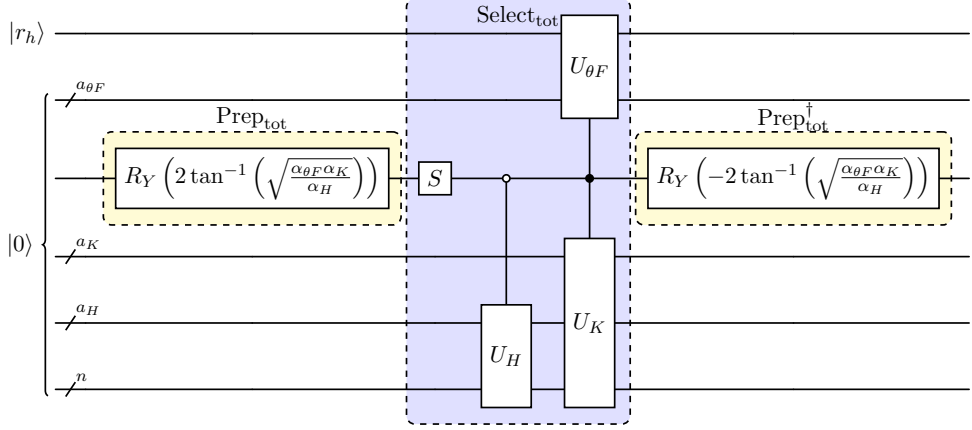
\begin{figure}[t]
    \begin{center}
    \scalebox{0.8}{
        \begin{quantikz}
        \lstick[1]{\ket{r_h}} & & & \gategroup[6,steps=3,style={dashed,rounded corners,fill=blue!12, inner xsep=2pt},background,label style={label position=above,anchor=north,yshift=-0.2cm}]{$\mathrm{Select}_{\mathrm{tot}}$} & & \gate[2]{U_{\theta F}} & &\\
        \lstick[5]{\ket{0}} & \qwbundle{a_{\theta F}} & & & & & &\\
        & & \gate{R_Y\left(2\tan^{-1}\left(\sqrt{\frac{\alpha_{\theta F} \alpha_K}{\alpha_H}}\right)\right)} \gategroup[1,steps=1,style={dashed,rounded corners,fill=yellow!20, inner xsep=2pt},background,label style={label position=above,anchor=north,yshift=0.33cm}]{$\mathrm{Prep}_{\mathrm{tot}}$} & \gate{S} & \ctrl[open]{2} & \ctrl{-1}\wire[d][1]{q} & \gate{R_Y\left(-2\tan^{-1}\left(\sqrt{\frac{\alpha_{\theta F} \alpha_K}{\alpha_H}}\right)\right)} \gategroup[1,steps=1,style={dashed,rounded corners,fill=yellow!20, inner xsep=2pt},background,label style={label position=above,anchor=north,yshift=0.42cm}]{$\mathrm{Prep}_{\mathrm{tot}}^\dagger$} &\\
        & \qwbundle{a_K} & & & & \gate[3]{U_K} & &\\
        & \qwbundle{a_H} & & & \gate[2]{U_H} & & &\\
        & \qwbundle{n} & & & & & &
        \end{quantikz}
        }
    \end{center}
\caption{\textbf{Block encoding of \(I\!\otimes\!H + i\,\theta F_h\!\otimes\!K\).}
An \(R_Y\) gate prepares amplitudes proportional to \(\sqrt{\alpha_H}\) and \(\sqrt{\alpha_{\theta F}\alpha_K}\); the branch unitary applies \(I\!\otimes\!U_H\) on \(|0\rangle\) and \(i\,U_{\theta F}\!\otimes\!U_K\) on \(|1\rangle\).
The overall unitary \(U_{\mathrm{tot}}\) is a \(\big(\alpha_H+\alpha_{\theta F}\alpha_K,\ 1+a_H+a_{\theta F}+a_K,\ \epsilon_{\mathrm{tot}}\big)\) block encoding of \(I\!\otimes\!H+i\,\theta F_h\!\otimes\!K\).}
\label{fig:block-encoding-total}
\end{figure}

Having obtained a block encoding of $I\otimes H + i\theta F_h \otimes K$, we can invoke standard Hamiltonian simulation algorithms to realize the time evolution. As these algorithms are well-established primitives, we do not elaborate on them further here.

\begin{algorithm}[t]
\caption{From block encodings of $H$, $K$, and $\theta F_h$ to a block encoding of $I\!\otimes\! H + i\,\theta F_h\!\otimes\! K$}
\label{alg:combine-and-simulate}
\begin{algorithmic}[1]
\State \textbf{Inputs:} 
$U_H$: $(\alpha_H,a_H,\epsilon_H)$ block encoding of $H$; \ 
$U_K$: $(\alpha_K,a_K,\epsilon_K)$ block encoding of $K$; \ 
$U_{\theta F}$: $(\alpha_{\theta F},a_{\theta F},\epsilon_{\theta F})$ block encoding of $\theta F_h$.
\State \textbf{Block encoding of $\theta F_h \otimes K$:} Construct $U_{\theta F\otimes K}:=U_{\theta F}\otimes U_K$, which is an $(\alpha_{\theta F}\alpha_K,\ a_{\theta F}{+}a_K,\ \epsilon_{\theta F\otimes K})$ block encoding of $\theta F_h\otimes K$ with $\epsilon_{\theta F\otimes K}=\epsilon_{\theta F}+\epsilon_K+\epsilon_{\theta F}\epsilon_K$.
\State \textbf{Combining $I\otimes H$ and $\theta F_h \otimes K$:} Apply $\mathrm{Prep}_{\text{tot}}=R_Y\!\big(2\arctan\sqrt{\alpha_{\theta F}\alpha_K/\alpha_H}\big)$, then $\mathrm{Select}_{\text{tot}}=\ket{0}\!\bra{0}\otimes(I\otimes U_H)\;+\;\ket{1}\!\bra{1}\otimes (i\,U_{\theta F\otimes K})$, and finally $\mathrm{Prep}_{\text{tot}}^\dagger$. The resulting unitary is $U_{\mathrm{tot}}=(\mathrm{Prep}_{\text{tot}}^\dagger\!\otimes I)\,\mathrm{Select}_{\text{tot}}\,(\mathrm{Prep}_{\text{tot}}\!\otimes I)$.
\State \textbf{Output:} $U_{\mathrm{tot}}$, an $(\alpha_{\mathrm{tot}},\ 1{+}a_H{+}a_{\theta F}{+}a_K,\ \epsilon_{\mathrm{tot}})$ block encoding of $I\!\otimes\! H + i\,\theta F_h\!\otimes\! K$, where $\alpha_{\mathrm{tot}}=\alpha_H+\alpha_{\theta F}\alpha_K$ and $\epsilon_{\mathrm{tot}}=\frac{\alpha_H}{\alpha_{\mathrm{tot}}}\,\epsilon_H+\frac{\alpha_{\theta F}\alpha_K}{\alpha_{\mathrm{tot}}}\,\epsilon_{\theta F\otimes K}$.
\end{algorithmic}
\end{algorithm}

\subsection{Evaluation step}
\label{sec:evaluate}

Finally, we describe the Evaluation step, where the ancilla is projected onto $\bra{l_h}$. The outcome of the simulation under the dilated Hamiltonian is
\[
U_E(T,0)\bigl(\ket{r_h}\otimes\ket{x_0}\bigr),
\]
where
\[
U_E(T,0)=\mathcal{T}\exp\!\left(-i\int_{0}^{T}\Big(I\otimes H(s)+i\theta F_h\otimes K(s)\Big)ds\right).
\]

As defined in Eq.~\eqref{eq:lh}, the evaluation functional is
\[
\bra{l_h}
= C_{M,\theta}\Bigl(\tfrac{M}{x}\Bigr)^\beta \bra{x},
\qquad x \in \Imid.
\]

In the language of quantum circuits, this corresponds to measuring the $m$ ancilla qubits introduced for the dilation and postselecting with outcome $\ket{x}$ for $x\in\Imid$. Since $M+1=2^m$, this is equivalent to postselecting when the two most significant bits of the $m$-qubit register are observed as $01$ or $10$. By applying amplitude amplification, one can boost the probability of obtaining such outcomes, and Theorem~\ref{thm:main} ensures that the system register is then prepared in a state close to
\[
\frac{\mathcal{T}e^{\int_0^T A(s)\,ds}\ket{x_0}}
{\bigl\|\mathcal{T}e^{\int_0^T A(s)\,ds}\ket{x_0}\bigr\|}.
\]

For amplitude amplification, we require two reflection operators~\cite{brassard2000quantum}. The first one is the oracle $S_\chi$, which flips the phase of states where the two most significant bits are $01$ or $10$, and acts trivially otherwise:
\[
S_\chi \ket{x} =
\begin{cases}
-\ket{x}, & \text{if the two most significant bits of $x$ are $01$ or $10$,}\\
\phantom{-}\ket{x}, & \text{otherwise}.
\end{cases}
\]
The second one is the reflection operator $S_0$, which flips the phase of the all-zero state:
\[
S_0 \ket{x} =
\begin{cases}
-\ket{0}, & \text{if } x=\ket{0},\\
\phantom{-}\ket{x}, & \text{otherwise}.
\end{cases}
\]

Here, $S_\chi$ can be implemented using two Pauli-$Z$ gates, while $S_0$ requires $O(m)$ single-qubit gates together with one $(m-1)$-bit Toffoli gate. The corresponding circuits are shown in Fig.~\ref{fig:reflections}.

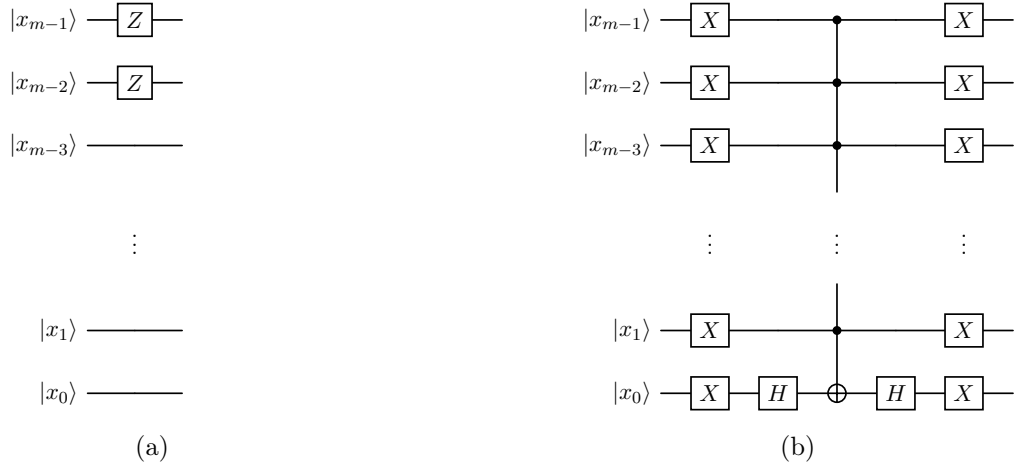
\begin{figure}[t]
    \begin{subfigure}{0.48\linewidth}
        \centering
        \scalebox{0.8}{
            \begin{quantikz}
                \lstick[1]{\ket{x_{m-1}}}& \gate{Z} &\\
                \lstick[1]{\ket{x_{m-2}}}& \gate{Z} &\\
                \lstick[1]{\ket{x_{m-3}}}& \ghost{Z} &\\
                \setwiretype{n} & & & & & &\\[0.001cm]
                \setwiretype{n} & \vdots &\\
                \setwiretype{n} & & & & & &\\[0.001cm]
                \lstick[1]{\ket{x_1}}& \ghost{Z} &\\
                \lstick[1]{\ket{x_0}}& \ghost{Z} &
            \end{quantikz}
            }
        \caption{}
    \end{subfigure}
    \hfill
    \begin{subfigure}{0.48\linewidth}
        \centering
        \scalebox{0.8}{
            \begin{quantikz}
                \lstick[1]{\ket{x_{m-1}}}& \gate{X} & & \control{} \wire[d][3]{q} & & \gate{X} &\\
                \lstick[1]{\ket{x_{m-2}}}& \gate{X} & & \control{} & & \gate{X} &\\
                \lstick[1]{\ket{x_{m-3}}}& \gate{X} & & \control{} & & \gate{X} &\\
                \setwiretype{n} & & & & & &\\[0.001cm]
                \setwiretype{n} & \vdots & & \vdots & & \vdots &\\
                \setwiretype{n} & & & & & &\\[0.001cm]
                \lstick[1]{\ket{x_1}}& \gate{X} & & \control{} \wire[u][1]{q}\wire[d][1]{q}& & \gate{X} &\\
                \lstick[1]{\ket{x_0}}& \gate{X} & \gate{H} & \targ{} & \gate{H} & \gate{X} &
            \end{quantikz}
            }
        \caption{}
    \end{subfigure}
    \caption{\textbf{Reflection operators for amplitude amplification.} 
    (a) Implementation of $S_\chi$ using two Pauli-$Z$ gates. 
    (b) Implementation of $S_0$, which requires $O(m)$ single-qubit gates and one $(m-1)$-bit Toffoli gate, decomposable into $2m-5$ standard Toffolis with an additional ancilla qubit.}
    \label{fig:reflections}
\end{figure}

\section{Summary and Discussion}

In this work we established a concrete pipeline that connects the mathematical dilation framework for linear differential equations with explicit quantum circuit realizations. On the analytical side, we showed that the discretized ancillary generator achieves skew-Hermitian structure on $\C^{M+1}$, allowing the moment conditions to be satisfied without resorting to artificial subspaces. By Theorem~\ref{thm:main}, if $\theta K_{\max}T \le 1/(8e)$ the global error admits the bound $\mathcal{O}(M^{-3/2}+M\,2^{-M/4})$, which decreases as the number of grid points $M+1$ increases. On the algorithmic side, we demonstrated that the triple $(F_h,|r_h\rangle,\langle l_h|)$ can be implemented efficiently on a gate-based quantum computer. This is achieved using linear combinations of unitaries with simple primitives such as QFT-adders, together with QSVT for the preparation of $\ket{r_h}$. Preparing the block encoding of $F_h$ requires only $\mathcal{O}((\log M)^2)$ gates, while the block encodings of $H$ and $K$ depend on the physical system under study. It should be noted, however, that the normalization factor scales as $\alpha_D=\mathcal{O}(M)$, so Hamiltonian simulation incurs a multiplicative overhead of $\mathcal{O}(M(\log M)^2)$. Thus the overall scheme captures the trade-off between circuit complexity and discretization error, enabling accurate simulation within the Hamiltonian simulation framework.

The broader implication of these results is that certain dissipative or open-system models, such as viscoelastic wave equations and other dissipative PDEs, can in principle be simulated within the Hamiltonian simulation framework when embedded via the proposed dilation. In this way, the methods developed here extend the applicability of quantum simulation beyond purely unitary dynamics, toward settings where loss and dissipation are intrinsic features.

At the same time, it is important to acknowledge some limitations of this work. The present analysis assumes that the dissipative part $K(t)$ is negative semidefinite; when this condition is violated, new instabilities and growth may arise, and understanding how the dilation behaves in such cases is an important open question. Moreover, while we derived resource counts and error guarantees, we did not perform explicit simulations; an immediate next step is the realization of small-scale circuits. Another promising direction is to investigate the use of higher-order SBP operators, which would tighten the error bounds, and to analyze how the increased stencil width affects circuit complexity when block-encoded on quantum hardware.

In summary, this paper highlights how dilation-based approaches, combined with quantum simulation algorithms, can make linear nonunitary dynamics accessible on quantum hardware. Looking ahead, further refinements and experimental validation are expected to bring such methods closer to practical applications in science and engineering.

\printbibliography

\end{document}